\documentclass[sigconf]{acmart}
\AtBeginDocument{%
  }

\usepackage{enumitem}

\setcopyright{acmlicensed}
\copyrightyear{2025}
\acmYear{2025}
\setcopyright{cc}
\setcctype{by}
\acmConference[ICAIF '25]{6th ACM International Conference on AI in Finance}{November 15--18, 2025}{Singapore, Singapore}
\acmBooktitle{6th ACM International Conference on AI in Finance (ICAIF '25), November 15--18, 2025, Singapore, Singapore}\acmDOI{10.1145/3768292.3771256}
\acmISBN{979-8-4007-2220-2/2025/11}

\begin{document}

\title{Fusing Narrative Semantics for Financial Volatility Forecasting}

\author{Yaxuan Kong}
\authornote{Both authors contributed equally to this research.}
\affiliation{%
  \institution{University of Oxford}
  \country{United Kingdom}
}
\email{yaxuan.kong@eng.ox.ac.uk}

\author{Yoontae Hwang}
\authornotemark[1]
\affiliation{%
  \institution{Pusan National University}
  \country{Republic of Korea}
}
\email{yoontae.hwang@pusan.ac.kr}

\author{Marcus Kaiser}
\affiliation{%
  \institution{Deutsche Bank AG}
  \city{London}
  \country{United Kingdom}}
\email{marcus.kaiser@db.com}

\author{Chris Vryonides}
\affiliation{%
  \institution{Deutsche Bank AG}
  \city{London}
  \country{United Kingdom}}
\email{chris.vryonides@db.com}

\author{Roel Oomen}
\affiliation{%
  \institution{Deutsche Bank AG}
  \city{London}
  \country{United Kingdom}}
\email{roel.oomen@db.com}

\author{Stefan Zohren}
\affiliation{%
  \institution{University of Oxford}
  \country{United Kingdom}
}
\email{stefan.zohren@eng.ox.ac.uk}


\begin{abstract}
  We introduce \texttt{M2VN}: Multi-Modal Volatility Network, a novel deep learning-based framework for financial volatility forecasting that unifies time series features with unstructured news data. \texttt{M2VN} leverages the representational power of deep neural networks to address two key challenges in this domain: (i) aligning and fusing heterogeneous data modalities, numerical financial data and textual information, and (ii) mitigating look-ahead bias that can undermine the validity of financial models. To achieve this, \texttt{M2VN} combines open-source market features with news embeddings generated by Time Machine GPT, a recently introduced point-in-time LLM, ensuring temporal integrity. An auxiliary alignment loss is introduced to enhance the integration of structured and unstructured data within the deep learning architecture. Extensive experiments demonstrate that \texttt{M2VN} consistently outperforms existing baselines, underscoring its practical value for risk management and financial decision-making in dynamic markets.
\end{abstract}



\maketitle

\section*{Disclaimer}
The opinions expressed in this article are those of the authors alone and do not necessarily represent the views of Deutsche Bank AG. This article is not intended to be comprehensive, nor does it constitute financial or other advice.

\section{Introduction}
Forecasting financial market volatility is a critical task in quantitative finance, with broad applications in risk management, derivative pricing, and portfolio optimization \citep{andersen2006volatility, alexander2008market}. Traditionally, volatility forecasting methods have relied primarily on features extracted directly from time series data, such as historical prices, returns, and various statistical measures \citep{clements2024harvesting}. However, recent advances in natural language processing (NLP) and the emergence of large language models (LLMs) have opened new possibilities for incorporating external textual information such as news articles into predictive models \citep{kong2024large, kong2024large_2}. While text-enhanced time series forecasting (TSF) has been explored in related domains \citep{liu2024time, jin2023time, gruver2023large}, research that specifically leverages LLM-derived textual features for volatility forecasting remains limited. 

Most existing studies incorporate textual data through word embeddings and use linear or NLP models for realized volatility (RV) forecasting \citep{rahimikia2021realised, parvini2025textual, atkins2018financial}. In contrast, the integration of rich, context-aware information extracted by LLMs with deep learning-based time series models in a unified framework is still largely underexplored. Deep learning models, with their ability to capture complex, nonlinear relationships and to integrate heterogeneous data sources, are particularly well-suited for modeling the intricate dynamics between financial time series and textual information. Nevertheless, progress in this area has been relatively limited to date.

\begin{figure}[!t]
\vspace{10pt}
  \centering
  \includegraphics[width=\linewidth]{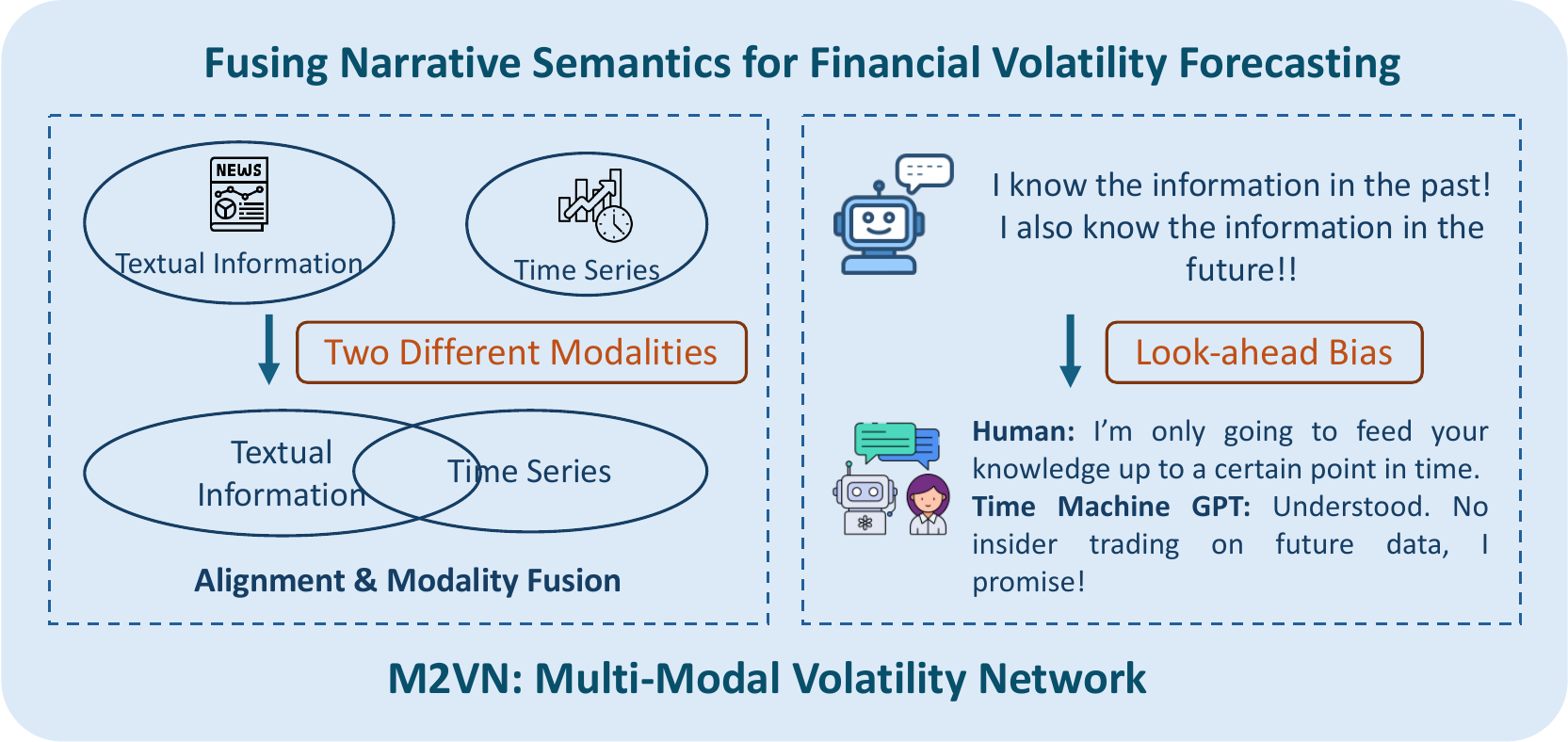}
  \Description{Illustration showing two challenges: data fusion across modalities and mitigation of look-ahead bias.}
  \caption{Illustration of two challenges: data fusion across modalities and mitigation of look-ahead bias.}
  \label{placeholder}
\vspace{-12pt}
\end{figure}

This limited exploration can largely be attributed to several inherent challenges unique to the intersection of financial time series and textual data (as illustrated in Figure 1). \textit{Firstly}, financial news and time series data represent distinct modalities, making their alignment and integration non-trivial. Specifically, learning how to effectively embed and fuse the intersection of textual information and numerical time series data into a unified embedding space is a significant challenge. \textit{Secondly}, utilizing news information with LLMs for financial prediction tasks often introduces look-ahead bias, whereby models inadvertently incorporate future knowledge, compromising their applicability and reliability. 

To address these challenges, we propose \textbf{M2VN: Multi-Modal Volatility Network}. \texttt{M2VN} innovatively combines traditional time series features - derived from open-sourced end-of-day price information (open, high, low, close, and volume) - with news data, where news content is embedded using Time Machine GPT (TiMaGPT) \citep{drinkall2024time}, a series of point-in-time language models specifically designed to mitigate look-ahead bias. To effectively align these two modalities, we introduce an auxiliary alignment loss function, which significantly enhances the model’s ability to coherently fuse textual and time series information within a unified framework. Through this design, \texttt{M2VN} can leverage both structured financial data and unstructured textual data to improve volatility forecasting performance. We validate the effectiveness of \texttt{M2VN} through extensive experiments, demonstrating its superior performance compared to existing baselines.

In summary, our paper makes the following key contributions:
\begin{enumerate}
    \item We propose \texttt{M2VN}: Multi-Modal Volatility Network, a novel framework to align and fuse news embeddings with time series data for the task of financial volatility forecasting.
    \item We assess the incremental value of incorporating rich news articles for RV forecasting, providing quantitative evidence for the benefits of augmenting time series models with contextual textual information.
    \item To the best of our knowledge, this is the first work to apply TiMaGPT \citep{drinkall2024time} for mitigating look-ahead bias in financial volatility forecasting with LLMs.
\end{enumerate}

\section{Related Works}
We review related works on financial volatility forecasting and text-enhanced TSF.
\subsection{Financial Volatility Forecasting} 
This section first reviews the main models used in volatility forecasting, then discusses recent advances in incorporating exogenous variables to enhance forecast accuracy.

\textbf{Typical Volatility Forecasting Models.}
Traditional volatility forecasting has been dominated GARCH \citep{engle1986modelling, andersen1998answering} models, which rely on daily squared returns to capture the dynamics of financial time series. While straightforward, these models are often hampered by the noisiness of daily return-based volatility proxies. The emergence of RV models, built on high-frequency intraday data, significantly improved predictive accuracy, especially with the introduction of the Heterogeneous Auto-Regressive (HAR) model \citep{corsi2009simple, degiannakis2017forecasting, audrino2020impact}. The HAR framework models current volatility using lagged daily, weekly, and monthly volatility, capturing volatility persistence in a simple linear structure \citep{clements2024harvesting}. Despite the advantages of high-frequency data, its limited accessibility and higher cost have led researchers to explore range-based estimators — such as Garman-Klass \citep{garman1980estimation}, Parkinson \citep{parkinson1980extreme}, and Rogers-Satchell \citep{rogers1991estimating} — that utilize daily open, high, low, and close prices. Recent empirical work \citep{lyocsa2021stock} demonstrates that, while high-frequency-based models outperform low-frequency ones for very short-term forecasts, the difference in forecast accuracy becomes negligible for longer horizons (e.g., one month), making low-frequency approaches both practical and effective in real-world settings.

\textbf{Exogenous Variables in Volatility Forecasting.}
To further enhance volatility forecasting performance, recent studies have introduced exogenous variables into HAR models, leading to the HAR-X framework \citep{degiannakis2017forecasting, clements2024harvesting}. These exogenous predictors typically include implied volatility indices, economic policy and equity market uncertainty indices, geopolitical risk measures, macroeconomic indicators, and asset-specific momentum. By introducing external information that is not captured by historical volatility alone, these additional features — many of which are available at daily frequencies — have been shown to enhance out-of-sample forecast accuracy \citep{korkusuz2024beyond}. Remarkably, several studies find that low-frequency HAR-X models can, in some cases, rival the forecasting performance of traditional high-frequency approaches \citep{clements2024harvesting}. Nevertheless, the literature \citep{clements2021practical, christensen2023machine} remains primarily focused on structured numerical data or linear models, with relatively few efforts devoted to combining unstructured text sources with these related time series features \citep{parvini2025textual, rahimikia2021realised}. The question of how to effectively align and fuse these disparate modalities remains largely unexplored, motivating the need for systematic research into multi-modal volatility forecasting.

\subsection{Text-Enhanced Time Series Forecasting}
This section reviews recent advances in integrating textual information into TSF, with a focus on applications of LLMs and the challenge of look-ahead bias in financial contexts.

\textbf{LLMs for TSF.}  
Effectively integrating external textual information has been shown to enhance TSF, for instance by correlating financial news sentiment with market movements \citep{wan2021sentiment}. More recently, research has shown that LLMs can significantly advance this approach. For instance, \citep{zhou2023one, hwang2025decision} demonstrated the versatility of LLMs across a range of time series tasks by employing a GPT-2 backbone capable of capturing temporal dependencies. Similarly, \citep{gruver2023large, ansari2024chronos} highlighted the zero-shot forecasting capabilities of pretrained LLMs, showing that appropriate tokenization allows these models to implicitly learn temporal patterns without the need for explicit task-specific training. Furthermore, \citep{jin2023time} proposed a reprogramming approach that converts time series data into formats more interpretable by LLMs, achieving state-of-the-art forecasting results. Their patch reprogramming method aligns time series patch embeddings with natural language representations, enhancing the model’s ability to understand and reason about time series data. Collectively, these studies underscore the promise of LLMs in advancing the field of TSF.

\textbf{Look-Ahead Bias in Financial Forecasting.} 
Despite their success in general domains, the application of LLM-based methods to financial time series is complicated by the risk of look-ahead bias. Look-ahead bias arises when models inadvertently access future information during training or evaluation, leading to unrealistically optimistic results and compromising the reliability of real-world forecasts \citep{sarkar2024lookahead}. To address this challenge, \citep{kim2024financial} propose using anonymized data to prevent LLMs from exploiting memorized future information; however, this approach may restrict the richness and specificity of textual data available to the model, potentially limiting its forecasting performance. To more robustly mitigate look-ahead bias while preserving the utility of external information, \citep{drinkall2024time} introduced TiMaGPT, a series of point-in-time LLMs specifically trained to maintain strict temporal integrity. By ensuring that models remain uninformed about future events or language changes during training, TiMaGPT provides a robust foundation for reliable TSF in dynamic financial contexts.

\section{Methodology}\label{method}
We propose \texttt{M2VN}: Multi-Modal Volatility Network, which is graphically illustrated in Figure 2. 

\begin{figure*}[!t]
\vspace{-5pt}
  \centering
  \includegraphics[width=0.92\linewidth]{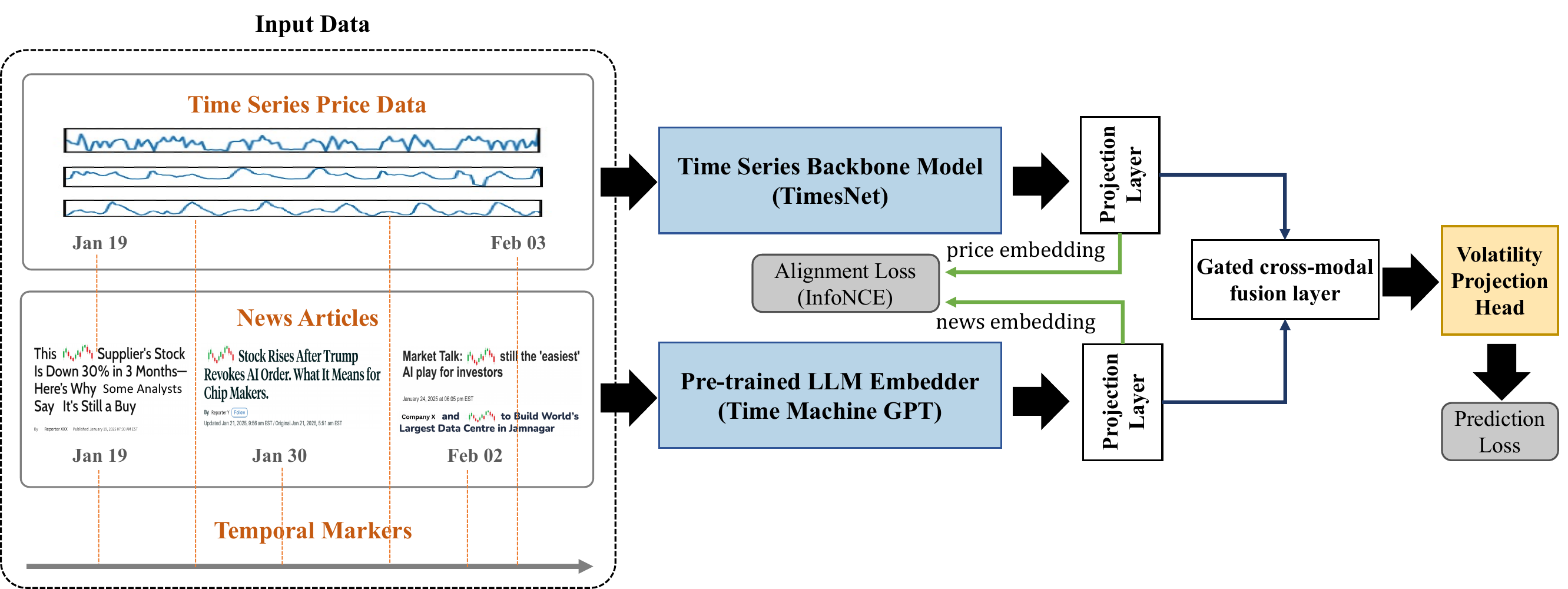}
  \caption{The architecture of M2VN, a multi-modal volatility network for financial volatility forecasting.}
  \label{placeholder}
\vspace{-8pt}
\end{figure*}

\subsection{Problem Formulation}
This subsection presents the formal problem formulation for multi-modal volatility forecasting and establishes the mathematical notation employed throughout our methodological exposition. We frame the task as a supervised, sequence-to-one learning problem, where the objective is to predict a future financial volatility measure by leveraging a historical look-back window of heterogeneous data streams. Our temporal domain consists of discrete trading days, indexed by $t=1, \dots, \mathcal{T}$. The core task is, at the close of day $t$, to predict the realized volatility for a future day $t+H$, where $H \ge 1$ is the prediction horizon. To achieve this, our model is provided with data spanning a look-back window of $T$ preceding days, covering the interval $[t-T+1, t]$. The input at each time $t$ is a multi-modal tuple $\mathbf{Z}_t = (X_t, N_t, M_t) \in \mathcal{X}$, where $\mathcal{X}$ is the input space. Each component of this tuple represents a distinct data modality.

First, $X_t = (\mathbf{x}_{t-T+1}, \dots, \mathbf{x}_t) \in \mathbb{R}^{T \times d_x}$ represents the \textbf{quantitative market state}. Each feature vector $\mathbf{x}_{\tau} \in \mathbb{R}^{d_x}$ is a concatenation of indicators designed to capture different facets of market dynamics, including historical volatility, price trends, and aggregate news information. Its composition is formally expressed as:

\begin{equation}
\small
\begin{split}
\mathbf{x}_{\tau} =\;[ &
  \underbrace{RV_{\tau}^{(d)}, RV_{\tau}^{(w)}, RV_{\tau}^{(m)}}_{\text{Realized Volatility}},\;
  \underbrace{M_{\tau}^{(w)}, M_{\tau}^{(m)}, M_{\tau}^{(q)}}_{\text{Momentum}},\\
  &\quad
  \underbrace{V_{\tau}^{(d)}}_{\text{Volume}},\;
  \underbrace{\text{VIX}_{\tau}^{(d)}}_{\text{VIX}},\;
  \underbrace{C_{\tau}^{(d)}}_{\text{News Count}}
]^{\top}
\end{split}
\end{equation}

Here, $RV_{\tau}^{(d)}$, $RV_{\tau}^{(w)}$, and $RV_{\tau}^{(m)}$ denote the daily, weekly (5-day average), and monthly (22-day average) RVs observed up to day $\tau$. The terms $M_{\tau}^{(w)}$, $M_{\tau}^{(m)}$, and $M_{\tau}^{(q)}$ are weekly, monthly, and quarterly momentum indicators. The trading volume is $V_{\tau}^{(d)}$. Finally, $\text{VIX}_{\tau}^{(d)}$ is the value of the VIX index, and $C_{\tau}^{(d)}$ is the count of relevant news articles for day $\tau$. 

Second, $N_t = (\mathbf{n}_{t-T+1}, \dots, \mathbf{n}_t) \in \mathbb{R}^{T \times d_n}$ constitutes the \textbf{qualitative information state}. Each vector $\mathbf{n}_{\tau} \in \mathbb{R}^{d_n}$ is a dense embedding derived from the titles and bodies of news articles published on day $\tau$, generated by a pre-trained transformer-based language model. The dimensionality $d_n$ corresponds to the output dimension of this language model. Finally, $M_t = (\mathbf{m}_{t-T+1}, \dots, \mathbf{m}_t) \in \mathbb{R}^{T \times d_m}$ contains \textbf{temporal markers}, where each $\mathbf{m}_{\tau} \in \mathbb{R}^{d_m}$ encodes calendar-related information (e.g., day of the week, month of the year) that captures known market seasonality. The prediction target, or ground truth, is a robust volatility estimator for day $t+H$, denoted as $y_{t+H} \in \mathcal{Y} \subseteq \mathbb{R}_{\ge 0}$. Following established financial literature, we define this target as an aggregation of multiple volatility estimators to yield a more stable and comprehensive measure. Specifically, let $V_{\tau}^{\text{P}}$, $V_{\tau}^{\text{GK}}$, and $V_{\tau}^{\text{RS}}$ represent the Parkinson \citep{parkinson1980extreme}, Garman-Klass \citep{garman1980estimation}, and Rogers-Satchell \citep{rogers1991estimating} volatility estimators for day $\tau$, respectively. The target variable is then constructed as $y_{t+H} = f_{\text{agg}}(V_{t+H}^{\text{P}}, V_{t+H}^{\text{GK}}, V_{t+H}^{\text{RS}})$, where $f_{\text{agg}}(\cdot)$ is a weighted averaging function \citep{clements2024harvesting}.

Our model is a learnable function $\Phi$ with parameters $\theta$, which maps an input instance $\mathbf{Z}_t$ from the input space $\mathcal{X}$ to a volatility prediction $\hat{y}_{t+H} \in \mathcal{Y}$. The learning process is governed by a dual-objective function, designed to not only ensure predictive accuracy but also to foster a meaningful semantic alignment between the quantitative and qualitative data streams. The primary objective is the \textbf{prediction loss}, which we define using the standard Mean Squared Error (MSE) to penalize deviations between the predicted and actual volatility:
\begin{equation}
\mathcal{L}_{\text{pred}}(\theta) = \mathbb{E}_{(\mathbf{Z}_t, y_{t+H})} [ ( \Phi(\mathbf{Z}_t; \theta) - y_{t+H} )^2 ]    
\end{equation}
To enhance the model's ability to fuse information, we introduce an auxiliary \textbf{alignment loss}, $\mathcal{L}_{\text{align}}$, based on the principles of contrastive learning. This objective encourages the intermediate representations of the market state and the news state for the same timestep within the look-back window to be similar, while being dissimilar to representations from other timesteps (cf.\citep{liu2024time}). Let $\mathbf{h}_{\tau}^{\text{quant}}$ and $\mathbf{h}_{\tau}^{\text{qual}}$ be the latent representations for the quantitative and qualitative modalities at timestep $\tau \in [t-T+1, t]$, respectively, derived from intermediate layers of $\Phi$. The alignment loss is formulated using an InfoNCE objective, which minimizes the cross-entropy between aligned pairs over a set of negative samples. The final optimization problem is thus to find the parameters $\theta^\star$ that minimize a weighted sum of these two objectives:
\begin{equation}
\small
\theta^\star = \arg\min_{\theta} \mathcal{L}(\theta) \quad \text{where} \quad \mathcal{L}(\theta) = \mathcal{L}_{\text{pred}}(\theta) + \lambda \mathcal{L}_{\text{align}}(\theta)
\end{equation}
Here, $\lambda > 0$ is a scalar hyperparameter that controls the relative importance of the alignment task. This multi-task learning paradigm is crucial for guiding the model to discover the complex, non-linear interactions between market dynamics and the narrative content of financial news.

The notations for the input data ($X_t, N_t, M_t$), the model function ($\Phi$), its parameters ($\theta$), and the dual-objective function ($\mathcal{L}_{\text{pred}}, \mathcal{L}_{\text{align}}$) will be used consistently in the following sections. We will now proceed to detail the specific architectural components of our proposed model, $\Phi$, explaining how its design is tailored to effectively address the learning problem defined herein.

\subsection{\texttt{M2VN}: Multi-Modal Volatility Network}

Forecasting the aggregated RV $y_{t+H}$ from the heterogeneous history $\mathbf Z_{t}=(X_{t},N_{t},M_{t})$ requires a representation that is simultaneously faithful to the statistical regularities of price trajectories, sensitive to the semantic content of financial narratives, and aware of calendar-driven seasonality.  The proposed architecture therefore processes the three modalities in a tightly coupled pipeline whose transformations preserve temporal coherence and promote cross-modal alignment, ultimately producing a single forecast $\hat y_{t+H}\in\mathcal Y$.  Formally, let $X_{t}\in\mathbb R^{T\times d_x}$, $N_{t}\in\mathbb R^{T\times d_n}$ and $M_{t}\in\mathbb R^{T\times d_m}$ denote the input tensors for a given prediction time $t$.  The network is parametrized by $\theta$ and is organized into exactly three interacting components that operate sequentially on these inputs.

\textbf{Multi-Modal Feature Encoder.}  The first component transforms each modality into a common latent space of width $d$ while retaining the full temporal resolution $T$.  Concretely, a learnable map $\mathcal E_{\text{price}}:\mathbb R^{T\times d_x}\to\mathbb R^{T\times d}$ embeds quantitative signals, using affine normalization followed by a periodic-spectral operator that decomposes the trajectory into its dominant frequency channels.  In parallel, a projection $\mathcal E_{\text{news}}^{\text{LLM}}:\mathbb R^{T\times d_n}\to\mathbb R^{T\times d}$ distills the transformer-derived news embeddings; its weights are shared across dates so that semantic drift, rather than scale variance, drives the dynamics. Finally, a harmonic embedding $\mathcal E_{\text{time}}:\mathbb R^{T\times d_m}\to\mathbb R^{T\times d}$ encodes calendar features \citep{zhou2021informer, wu2021autoformer}. The encoder outputs three sequences
\begin{equation}
\small
\mathbf h^{\text{price}}_{1:T}=\mathcal E_{\text{price}}(X_t) \quad
\mathbf h^{\text{news}}_{1:T}=\mathcal E_{\text{news}}^{\text{LLM}}(N_t) \quad
\mathbf h^{\text{time}}_{1:T}=\mathcal E_{\text{time}}(M_t)    
\end{equation}
and each belonging to $(\mathbb R^{d})^{T}$.  For brevity, we denote their joint concatenation as $\mathbf h_{1:T}\in\mathbb R^{T\times 3d}$. This tensor constitutes the raw material on which temporal reasoning is performed.

Note that employing a modern language model such as GPT-4.1 to embed the news stream $N_t$ can spuriously inflate forecasting accuracy, because the model’s training corpus may already contain future articles or hindsight knowledge. To neutralize this look-ahead bias, we adopt TiMaGPT \citep{drinkall2024time}. For every date $\tau$, the embedding $\mathbf n_\tau$ is generated with TiMaGPT weights whose knowledge cut-off has been rolled back to exactly one year before $\tau$. In practice, we advance the cutoff year-by-year during training, freeze the corresponding parameters, and feed the day’s news content — obtained by concatenating all articles published on $\tau$ — through this frozen model to obtain a temporally faithful vector representation.

\textbf{Latent Dynamics Module.} The second component models long-range dependencies and cross-modal interactions over the horizon $T+H$ without resorting to causal truncation \citep{wu2023timesnet}.  Let $\mathbf h^{(0)}_{1:T}=\mathbf h_{1:T}$.  A stack of $L$ identical blocks iteratively refines these embeddings, each block $\mathcal T_\ell:(\mathbb R^{T\times 3d})\to(\mathbb R^{T\times 3d})$ being defined as
\begin{equation}
\mathcal T_\ell(\mathbf h) = \mathbf h + \mathcal G_\ell(\mathcal C_\ell(\mathbf h)),
\end{equation}
where $\mathcal C_\ell$ is a spectral–convolution operator that exploits the discrete Fourier transform to isolate the $k$ most energetic periods of the price trajectory before applying an inception-style two-dimensional convolution across the period–channel grid. 

A key difficulty in volatility forecasting is teasing apart informative, low-frequency rhythms (e.g. weekly or monthly cycles) from high-frequency stochastic noise. To address this, proceeds in two steps (See \citep{wu2023timesnet}):

\begin{itemize}[left=6pt]
    \item \textbf{Spectral filtering.} Transform the latent trajectory to the Fourier domain and keep the modes with the highest energy, isolating the most structurally significant components.
    \item \textbf{Inception-style 2D convolution.} Apply multi-scale kernels to the filtered signal to capture residual cross-temporal interactions.
\end{itemize}

The choice to select the top energy modes in the backbone is a principled decision, not a heuristic. As proven in Proposition~\ref{prop:best-k-term}, this yields the best k-term approximation that minimizes the Frobenius reconstruction error. This approach is a well-established principle in fields such as Fourier analysis, nonlinear approximation, and compressed sensing, where it has been shown that selecting the k largest transform coefficients provides an optimal approximation \citep{cohen2006compressed, davenport2012introduction}. Consequently, this method imposes a strong inductive bias that encourages the network to focus on robust, periodic structures rather than overfitting to incidental noise. This built-in denoising mechanism can materially improve generalization and stability.

It is instructive to provide the following result, the proof of which follows from standard results \citep{cohen2006compressed, davenport2012introduction}:


\begin{proposition}[Best $k$-term Fourier approximation]\label{prop:best-k-term}
Let $\mathbf h=(h_0,\dots,h_{T-1})\in(\mathbb C^{d})^{T}$ and denote by
$\widehat{\mathbf h}=\mathcal F\mathbf h$ its unitary discrete
Fourier transform
\begin{equation}
  \widehat{\mathbf h}_{\omega}
  =
  \frac{1}{\sqrt{T}}
  \sum_{t=0}^{T-1}h_t\,e^{-2\pi i\omega t/T},
  \qquad \omega=0,\dots,T-1 .    
\end{equation}
For $1 \le k \le T$ let
$I_k\subseteq\{0,\dots,T-1\}$ collect the $k$ indices of largest
energy $\|\widehat{\mathbf h}_{\omega}\|_2$ (ties broken arbitrarily) and set
\begin{equation}
  (\mathcal P_k\widehat{\mathbf h})_{\omega}
  =\begin{cases}
      \widehat{\mathbf h}_{\omega}, & \omega\in I_k,\\
      0, & \omega\notin I_k .
    \end{cases}
  \qquad
  \mathbf h_k  = \mathcal F^{-1}(\mathcal P_k\widehat{\mathbf h}).    
\end{equation}
Then $\mathbf h_k$ minimises the approximation error
\begin{equation}
  \mathbf h_k
  =
  \operatorname*{arg\,min}_{\tilde{\mathbf h}\in(\mathbb C^{d})^{T}}
  \{
     \|\mathbf h-\tilde{\mathbf h}\|_F^2
     \enspace | \enspace 
     \|\mathcal F\tilde{\mathbf h}\|_0\le k
  \},    
\end{equation}
where $\|\cdot\|_F$ is the Frobenius norm and
$\|\cdot\|_0$ counts non-zero Fourier coefficients.
\end{proposition}

\begin{proof}
Parseval’s identity for the unitary DFT gives
\begin{equation}\label{eq:parseval}
  \|\mathbf h-\tilde{\mathbf h}\|_F^2 = \sum_{\omega=0}^{T-1}
      \|
        \widehat{\mathbf h}_{\omega}
        -\widehat{\tilde{\mathbf h}}_{\omega}
      \|_2^2
  \qquad
  (\forall\tilde{\mathbf h}\in(\mathbb C^{d})^{T}).
\end{equation} If a candidate $\tilde{\mathbf h}$ has spectral support $S=\{\omega:\widehat{\tilde{\mathbf h}}_{\omega}\neq0\}$ with $|S|\le k$, the summand for $\omega\in S$ in \eqref{eq:parseval} is minimised by $\widehat{\tilde{\mathbf h}}_{\omega}=\widehat{\mathbf h}_{\omega}$, while the optimal choice for $\omega\notin S$ is $\widehat{\tilde{\mathbf h}}_{\omega}=0$. Consequently the minimal achievable error for this support equals
\begin{equation}
  E(S)=  \sum_{\omega\notin S}\|\widehat{\mathbf h}_{\omega}\|_2^2 .
\end{equation}
Because $E(S)$ decreases when $\sum_{\omega\in S}\|\widehat{\mathbf h}_{\omega}\|_2^2$ increases, a support of cardinality $k$ is optimal only when it captures the $k$ largest energies — namely $S=I_k$ - yielding the value $\sum_{\omega\notin I_k}\|\widehat{\mathbf h}_{\omega}\|_2^2$. Reconstructing via the inverse DFT produces the minimiser $\mathbf h_k$ declared above.
\end{proof}

And then, the function $\mathcal G_\ell$ performs gated cross-modal fusion. Given the price representation $\mathbf r_{1:T}$ and the news representation $\mathbf t_{1:T}$ isolated from the input of the block, it computes a gate $a_{1:T}=\sigma(W_g[\mathbf r_{1:T};\mathbf t_{1:T}])\in(0,1)^{T\times1}$ and returns
\begin{equation}
\mathbf z_{1:T}=a_{1:T}\odot\mathbf r_{1:T} + (1-a_{1:T})\odot\mathbf t_{1:T},    
\end{equation}
together with bilinear and absolute-difference interactions, concatenated and linearly projected back to $\mathbb R^{d}$.  The entire sequence $\mathbf h^{(L)}_{1:T}$ emerging from the final block therefore contains price-aware news semantics and news-aware price spectra.  To encourage semantic consistency across modalities, we extract aligned embeddings
\begin{equation}
\mathbf r_{1:T}=P_{\text{price}}(\mathbf h^{(L)}_{1:T}),\qquad 
\mathbf t_{1:T}=P_{\text{news}}(\mathbf h^{(L)}_{1:T}),    
\end{equation}
where $P_{\ast}:(\mathbb R^{T\times 3d})\to(\mathbb R^{T\times d_a})$ are linear projections into an alignment space of width $d_a$.  The InfoNCE loss \citep{oord2018representation} introduced in the Preliminary subsection is then realized as
\begin{equation}
\mathcal L_{\text{align}}(\theta)= \mathbb E[-\log\frac{\exp(\langle \mathbf r_{\tau},\mathbf t_{\tau}\rangle/\tau)} {\sum_{\tau'}\exp(\langle \mathbf r_{\tau},\mathbf t_{\tau'}\rangle/\tau)}],    
\end{equation}
where $\langle\cdot,\cdot\rangle$ denotes the inner product in $\mathbb R^{d_a}$ and $\tau>0$ is a learned temperature parameter.

\textbf{Volatility Projection Head.} The final component converts the fused latent sequence into a point prediction $ \hat y_{t+H}$.  A temporal projection $W_p\in\mathbb R^{(T+H)\times T}$ is applied to the price-channel slice of $\mathbf h^{(L)}_{1:T}$, yielding an extended sequence $\tilde{\mathbf h}_{1:T+H}\in\mathbb R^{(T+H)\times d}$ that extrapolates the learned dynamics up to the horizon $H$.  We retain the last position and pass it through an affine map $W_o:\mathbb R^{d}\to\mathcal Y$, obtaining
\begin{equation}
\hat y_{t+H}=W_o(\tilde{\mathbf h}_{T+H}).    
\end{equation}

A deterministic output is appropriate because the target $y_{t+H}$ is defined as a smoothed aggregation of three volatility estimators and therefore concentrates most of its uncertainty in the inputs rather than in the measurement.

The network parameters are trained by minimizing the joint objective
\begin{equation}
\mathcal J(\theta) =\mathbb E_{(\mathbf Z_t,y_{t+H})} [(\Phi(\mathbf Z_t;\theta)-y_{t+H})^2] +\lambda\,\mathcal L_{\text{align}}(\theta),    
\end{equation}
where the first term is the mean-squared prediction loss $\mathcal L_{\text{pred}}$ defined earlier and $\lambda>0$ controls the relative weight of cross-modal alignment.  Stochastic gradient descent with automatic differentiation is employed to solve $\theta^\star=\arg\min_\theta\mathcal J(\theta)$.

\begin{table*}[htbp]
\centering
\small
\renewcommand{\arraystretch}{1.2}
\begin{tabular}{lccccccccccccccc}
\toprule
\multicolumn{1}{l}{\textbf{Method}} &
\multicolumn{2}{c}{\textbf{KO}} &
\multicolumn{2}{c}{\textbf{CMCSA}} &
\multicolumn{2}{c}{\textbf{COP}} &
\multicolumn{2}{c}{\textbf{GILD}} &
\multicolumn{2}{c}{\textbf{MRK}} &
\multicolumn{2}{c}{\textbf{NKE}} &
\multicolumn{2}{c}{\textbf{ORCL}} \\
\cmidrule(lr){2-3}\cmidrule(lr){4-5}\cmidrule(lr){6-7}
\cmidrule(lr){8-9}\cmidrule(lr){10-11}\cmidrule(lr){12-13}\cmidrule(lr){14-15}
& QLike & MAPE & QLike & MAPE & QLike & MAPE & QLike & MAPE &
  QLike & MAPE & QLike & MAPE & QLike & MAPE \\ \midrule
HAR                & 0.0775 & 0.2932 & 0.0937 & 0.3317 & 0.0749 & 0.2982 & 0.0674 & 0.3290 & 0.0804 & 0.3106 & 0.0766 & 0.3179 & 0.0972 & 0.3219 \\
HAR-X (OLS)        & 0.0658 & 0.3507 & 0.0746 & 0.3317 & 0.0779 & 0.2889 & 0.0961 & 0.4447 & 0.1060 & 0.3705 & 0.0889 & 0.2946 & 0.1068 & 0.3185 \\
HAR-X (Ridge)      & 0.0638 & 0.3484 & \textcolor{blue}{\textbf{0.0723}} & 0.3388 & 0.0746 & 0.2873 & 0.0764 & 0.3942 & 0.0762 & 0.3361 & 0.0713 & 0.2946 & 0.0926 & \textcolor{blue}{\textbf{0.3082}} \\
HAR-X (Lasso)      & \textcolor{blue}{\textbf{0.0627}} & 0.3535 & \textcolor{red}{\textbf{0.0692}} & 0.3216 & 0.0594 & \textcolor{blue}{\textbf{0.2818}} & 0.0716 & 0.4043 & \textcolor{blue}{\textbf{0.0685}} & 0.3475 & \textcolor{blue}{\textbf{0.0646}} & 0.2873 & \textcolor{blue}{\textbf{0.0809}} & \textcolor{red}{\textbf{0.3047}} \\
Informer           & 0.1629 & 0.4028 & 0.1125 & 0.3403 & 0.1399 & 0.3519 & 0.1007 & 0.4035 & 0.1959 & 0.3311 & 0.1453 & 0.3017 & 0.2443 & 0.4125 \\
Autoformer         & 0.1095 & 0.3756 & 0.1558 & 0.4036 & 0.1480 & 0.3625 & 0.1089 & 0.3622 & 0.0906 & 0.3420 & 0.1120 & 0.3297 & 0.2119 & 0.4302 \\
DLinear            & 0.0845 & 0.3371 & 0.0928 & 0.3355 & 0.0665 & 0.3083 & 0.0688 & 0.2977 & 0.0792 & 0.3174 & 0.0744 & 0.3103 & 0.1006 & 0.3614 \\
TimesNet           & 0.0666 & \textcolor{blue}{\textbf{0.2943}} & 0.0833 & 0.3177 & 0.0652 & 0.2839 & \textcolor{blue}{\textbf{0.0660}} & 0.2852 & 0.0760 & 0.3065 & 0.0679 & 0.2905 & 0.0864 & 0.3182 \\
PAttn              & 0.0699 & 0.3046 & 0.0805 & \textcolor{blue}{\textbf{0.3127}} & 0.0603 & 0.2888 & 0.0678 & 0.2853 & 0.0699 & 0.3031 & 0.0695 & 0.3038 & 0.0959 & 0.3099 \\
TimeXer            & 0.0661 & 0.3040 & 0.0807 & 0.3162 & \textcolor{blue}{\textbf{0.0589}} & 0.3035 & 0.0673 & \textcolor{red}{\textbf{0.2760}} & 0.0797 & \textcolor{red}{\textbf{0.2909}} & 0.0671 & \textcolor{blue}{\textbf{0.2865}} & 0.0849 & 0.3255 \\
\textbf{M2VN (Ours)} & \textcolor{red}{\textbf{0.0589}} & \textcolor{red}{\textbf{0.2771}} & 0.0769 & \textcolor{red}{\textbf{0.3025}} & \textcolor{red}{\textbf{0.0552}} & \textcolor{red}{\textbf{0.2722}} & \textcolor{red}{\textbf{0.0630}} & \textcolor{blue}{\textbf{0.2838}} & \textcolor{red}{\textbf{0.0679}} & \textcolor{blue}{\textbf{0.3005}} & \textcolor{red}{\textbf{0.0594}} & \textcolor{red}{\textbf{0.2747}} & \textcolor{red}{\textbf{0.0696}} & 0.3092 \\
\bottomrule
\end{tabular}
\vspace{2mm}

\caption{Out-of-sample forecasting accuracy of ten statistical and deep learning models for daily volatility of seven U.S. stocks. Results are for quasi-likelihood loss (QLike) and mean absolute percentage error (MAPE). For each stock–metric pair, the best model is in \textcolor{red}{red}, second-best in \textcolor{blue}{blue}. Values are averages of three runs; standard deviations are sufficiently small and omitted.}
\label{tab1:main_result}
\vspace{-0.5cm}
\end{table*}

\begin{table*}[htbp]

\centering
\small
\renewcommand{\arraystretch}{1.2}
\begin{tabular}{lcccccccccccccc}
\toprule
\multicolumn{1}{l}{\textbf{Ablation}} & \multicolumn{2}{c}{\textbf{KO}} & \multicolumn{2}{c}{\textbf{CMCSA}} & \multicolumn{2}{c}{\textbf{COP}} & \multicolumn{2}{c}{\textbf{GILD}} & \multicolumn{2}{c}{\textbf{MRK}} & \multicolumn{2}{c}{\textbf{NKE}} & \multicolumn{2}{c}{\textbf{ORCL}} \\
\cmidrule(lr){2-3} \cmidrule(lr){4-5} \cmidrule(lr){6-7} \cmidrule(lr){8-9} \cmidrule(lr){10-11} \cmidrule(lr){12-13} \cmidrule(lr){14-15}
& QLike & MAPE & QLike & MAPE & QLike & MAPE & QLike & MAPE & QLike & MAPE & QLike & MAPE & QLike & MAPE \\
\midrule
M2VN & 
\textbf{0.0589} & \textbf{0.2771} & \textbf{0.0769} & \textbf{0.3025} & \textbf{0.0552} & \textbf{0.2722} & \textbf{0.0630} & \textbf{0.2838} & \textbf{0.0679} & \textbf{0.3005} & \textbf{0.0594} & \textbf{0.2747} & \textbf{0.0696} & 0.3092 \\
M2VN \! {\scriptsize W/O $V_{\tau}^{(d)}$}\!\! \!\! \!\! & 
0.0610 & 0.2783 & 0.0812 & 0.3105 & 0.0580 & 0.2822 & 0.0659 & 0.2877 & 0.0701 & 0.3100 & 0.0648 & 0.3039 & 0.0774 & \textbf{0.3076} \\ 
M2VN \! {\scriptsize W/O $N_{\tau}$}\!\! \!\! \!\!       & 0.0641 & 0.2904 & 0.0866 & 0.3184 & 0.0599 & 0.3000 & 0.0632 & 0.2851 & 0.0692 & 0.3169 & 0.0638 & 0.2912 & 0.0780 & 0.3251 \\
\bottomrule
\end{tabular}

\vspace{2mm}
\caption{Ablation study on the impact of the daily volume feature ($V_{\tau}^{(d)}$) in \texttt{M2VN} forecasting, as well as impact of news modality on \texttt{M2VN} volatility forecasting performance. Bold indicates the best result for each stock-metric pair.}
\label{tab1:main_ablation}
\vspace{-0.5cm}
\end{table*}

\section{Experiments}
Now we present experiment results to thoroughly demonstrate the performance of \texttt{M2VN} on real-world benchmark datasets. The pre-processed dataset, raw news article data, and implementation code are publicly accessible at https://github.com/Yoontae6719/M2VN-Multi-Modal-Learning-Network-for-Volatility-Forecasting

\subsection{Implementation Details}
We present the details of datasets, baseline models, training details and evaluation metrics.

\textbf{Data Sources.} We construct a unified daily panel comprising seven large-capitalization U.S. equities — Coca-Cola (KO), Comcast (CMCSA), ConocoPhillips (COP), Gilead Sciences (GILD), Merck \& Co. (MRK), Nike (NKE) and Oracle (ORCL). For each ticker, we extract both end-of-day price information (open, high, low, close and volume) and the full text of associated news articles from FNSPID \citep{dong2024fnspid}. The FNSPID dataset \citep{dong2024fnspid} is an open-source resource that integrates financial news articles with corresponding time series data. All news articles are temporally aligned on a daily basis in a structured format. Owing to the limited availability of news data, above seven tickers were selected according to the volume of associated articles.

\textbf{Data Processing.}
Both news and time series data are partitioned chronologically into a training span from 1 January 2013 to 31 December 2017, a validation span from 1 January 2018 to 31 December 2020, and a strictly held-out test span from 1 January 2021 to 15 December 2023. News documents are temporally aligned so that only articles published on the calendar day before as the corresponding price record enter the model’s input for that day, thereby precluding any forward-looking leakage.

\textbf{Baseline models.} To evaluate the contribution of \texttt{M2VN}, we compare it against ten strong baselines spanning both classical econometrics and contemporary deep learning. The classical group contains Standard \textbf{HAR}, \textbf{HAR-X (OLS)}, \textbf{HAR-X (Lasso)} and \textbf{HAR-X (Ridge)}, variants of the Heterogeneous AutoRegressive model widely used in volatility forecasting \citep{clements2024harvesting}. The modern deep learning group encompasses \textbf{Informer}\citep{zhou2021informer}, \textbf{Autoformer}\citep{wu2021autoformer}, \textbf{DLinear}\citep{zeng2023transformers}, \textbf{TimesNet}\citep{wu2023timesnet}, \textbf{TimeXer}\citep{wang2024timexer} and the \textbf{PAttn} \citep{tan2024language} architecture. For all three HAR-X baselines, we first reduce the dimensionality of the daily news embeddings via principal-component analysis (PCA), retaining the top components that explain 95\% of variance. This mitigates the curse of dimensionality while keeping the econometric specification parsimonious. All deep learning baselines ingest the raw news embeddings directly.

\textbf{Training detail.} The model is trained using a sequence length of 12 time steps, where the last 2 steps serve as contextual labels, and it forecasts 1 step into the future. All deep learning baselines and our \texttt{M2VN} model are optimized with Adam, an initial learning rate of $3\times10^{-4}$, and cosine decay. Early stopping on the validation set prevents over-fitting. In our models, hyper-parameters are selected by random search, drawing 100 configurations from the following grids: latent dimension $d\in \{12,24,32,64\}$, alignment dimension $d_{\alpha}\in \{32,64,128,256\}$, temperature $\tau\in \{0.01,0.03,0.07\}$, Inception hidden width $d_{i}\in \{24,32,64,128,256,512\}$ and top-$k$ signal $k \in \{4, 5, 6\}$ . Equivalent searches (with architecture-specific ranges) are conducted for every deep baseline to ensure a fair comparison. Each experiment is repeated with three random seeds. we report the mean across runs.

\textbf{Evaluation metrics.} Model performance is quantified with Mean Absolute Percentage Error (MAPE) and the Quasi-Likelihood loss (QLIKE). MAPE delivers an easily interpretable, scale-free percentage error that treats over- and under-predictions symmetrically. QLIKE, rooted in the quasi-log-likelihood of a Gaussian volatility model, places extra weight on under-forecasting during high-volatility episodes. For a forecast $\hat{y}_{t}$ of the realised volatility $y_{t}$ over
$n$ test points, they are defined as
\begin{equation}
\small
\text{MAPE} = \frac{100\%}{n}\sum_{t=1}^{n}|\frac{y_t-\hat{y}_t}{y_t}|,\quad
\end{equation}

\begin{equation}
\small
\text{QLIKE} =\frac{1}{n} \sum_{i=1}^{n} \left(\frac{y_i}{\hat{y}_i}- \ln\!\left(\frac{y_i}{\hat{y}_i}\right) - 1\right).
\end{equation}
As our model directly forecasts volatility, we define both $y_t$ and $\hat{y}$ in terms of standard deviation for the QLIKE calculation. While the QLIKE metric is conventionally applied to variance, this consistent application to standard deviation allows it to serve as a robust tool for the internal ranking of model performance. It preserves the essential characteristic of penalizing under-prediction more heavily, which is crucial for our comparative study.

\begin{figure*}[!htbp]
\small{
  \centering
  \includegraphics[width=\linewidth]{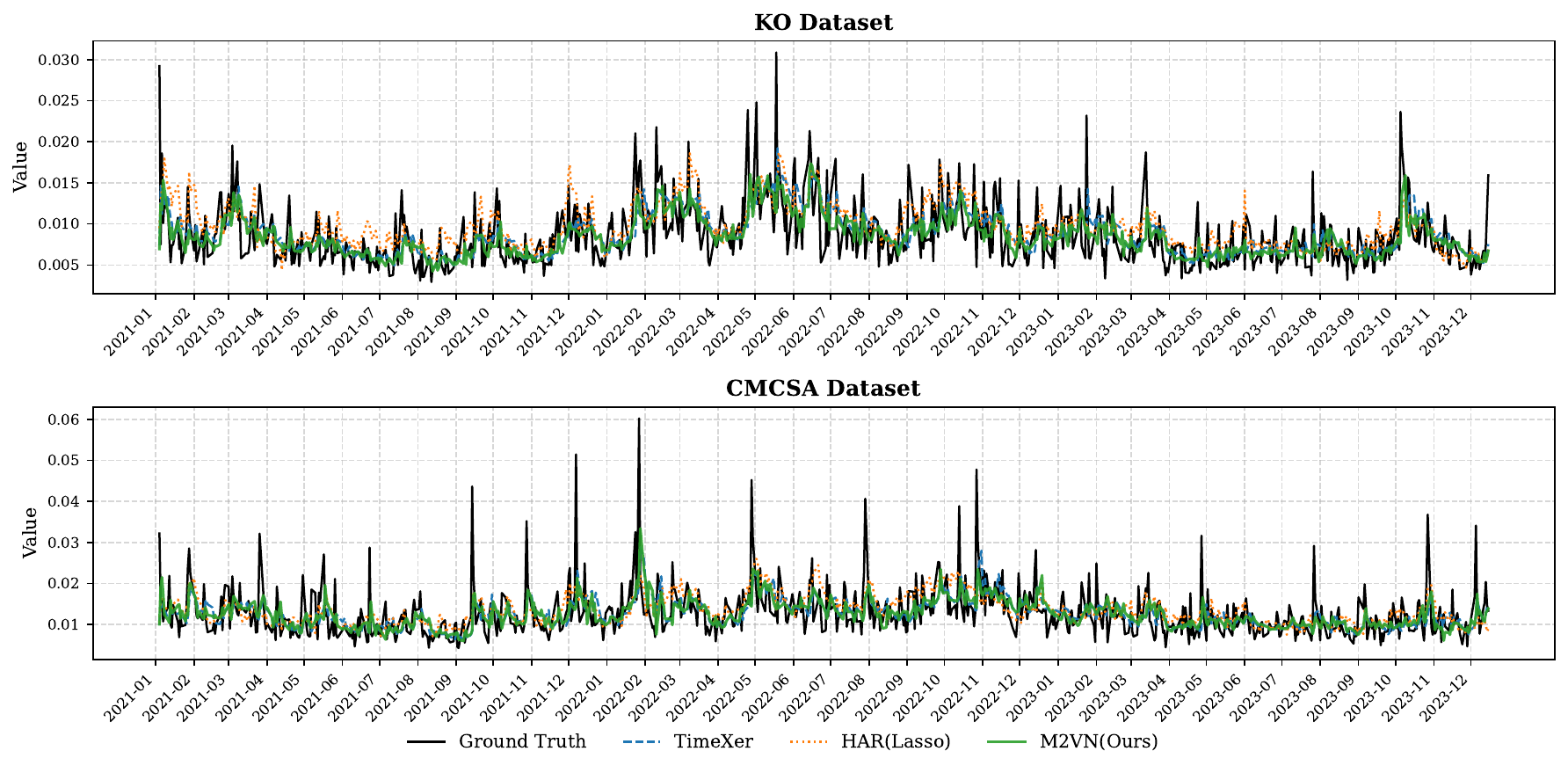}}
  \caption{Visual comparison of out-of-sample forecast quality for the KO (top) and CMCSA (bottom) datasets over the test period of 2021-2023 . The \texttt{M2VN} model's predictions (\textcolor{green}{green}) are plotted against the ground truth realized volatility (black), alongside forecasts from the HAR-X (Lasso) (\textcolor{orange}{orange}) and TimeXer (\textcolor{blue}{blue}) baselines. The \texttt{M2VN} demonstrates a superior ability to track the fluctuations and sharp peaks of the ground truth data compared to the other models.}
  \label{figure_showcase}
\end{figure*}

\subsection{\texttt{M2VN} is Useful for Volatility Forecasting?}
To answer the central research question of whether our proposed multi-modal framework improves volatility forecasting, we conduct a comprehensive empirical evaluation against a suite of strong classical and deep learning baselines. This section details the out-of-sample performance of \texttt{M2VN}, whose architecture and learning objectives are described in Section \ref{method}, on seven real-world equity datasets.

The quantitative results are presented in Table \ref{tab1:main_result}. Our proposed model, \texttt{M2VN}, demonstrates consistently strong performance, outperforming eight out of ten baselines in the majority of stock-metric settings. Specifically, \texttt{M2VN} achieves either the best (highlighted in \textcolor{red}{red}) or second-best (in \textcolor{blue}{blue}) result in 12 out of 14 evaluation settings (across 7 stocks and 2 metrics). Notably, it secures the top rank in ten cases (71.4\%) and places within the top two in the vast majority of settings. While the HAR-X (Lasso) baseline yields some competitive results, obtaining the best performance in two settings (CMCSA–QLike and ORCL–MAPE), \texttt{M2VN} either matches or exceeds Lasso’s accuracy in most other stock–metric combinations, highlighting its overall superiority. Relative to \textit{TimesNet} \citep{wu2023timesnet}, an ablation of our backbone that omits the cross-modal alignment loss $\mathcal{L}_{\text{align}}$, mean QLike and MAPE fall by 11.7 \% and 3.6 \%, respectively. 

The superior performance of \texttt{M2VN} can be attributed to several key architectural innovations. First, the explicit fusion of quantitative market data and qualitative news narratives, enforced by the auxiliary contrastive alignment loss ($\mathcal{L}_{\text{align}}$), enables the model to learn a richer, more robust representation of the underlying market state than models relying on price data alone. Second, the latent dynamics module, which leverages spectral filtering to isolate the dominant periodic components of the volatility signal (cf. Proposition \ref{prop:best-k-term}), effectively denoised the input sequences. This structural inductive bias may encourages the model to focus on persistent, predictable patterns rather than overfitting to stochastic noise, a common pitfall in financial TSF. These empirical results suggest that the architectural choices embodied in \texttt{M2VN} are suitable for the complex task of volatility forecasting.

\subsection{Is Volume a Proxy for News Information?}
To investigate the intricate relationship between market activity and narrative information, we examine whether trading volume can serve as a proxy for news. According to the efficient market hypothesis, news is rapidly incorporated into prices, with trading volume often reflecting the arrival and processing of such information. This raises the question of whether the information contained in news is subsumed by the more readily available quantitative volume data.

To test this hypothesis, we conduct a targeted ablation study on our \texttt{M2VN} model. Specifically, we compare the performance of the full model — which leverages both daily trading volume ($V_{\tau}^{(d)}$) from the quantitative market state ($X_{t}$) and news embeddings ($N_{t}$) - against an identical model with the volume feature ablated. In this alternative configuration, the model relies on news embeddings and other market variables to capture information flow. If trading volume were merely a proxy for news, removing it from a model already equipped with rich news representations should yield only marginal changes in performance.

The results of this experiment are summarized in Table \ref{tab1:main_ablation}. Excluding trading volume results in a statistically significant and consistent reduction in forecasting accuracy across all seven equities. The complete model (with volume, standard \texttt{M2VN}) outperforms its ablated counterpart (without volume, \texttt{M2VN} W/O $V_{\tau}^{(d)}$) on 27 out of 28 stock-metric pairs. On average, incorporating volume reduces MAPE by 2.86$\pm$3.06\% and QLike by 5.64$\pm$2.40\% across all stocks.

These findings suggest that trading volume is not merely a proxy for news but rather a source of \textbf{complementary information}. The substantial improvement in predictive accuracy achieved by including volume — even alongside rich textual news features — indicates that volume captures unique aspects of market dynamics. These may capture investor disagreement, liquidity, or the intensity of market attention, which are not fully reflected in the semantic content of news alone. Therefore, to develop high-fidelity financial forecasting models, it is crucial to regard both news narratives and trading volume as essential, non-redundant sources of information.

\subsection{Is News Embedding Predictive?}
To ascertain the direct contribution of news embeddings to the predictive capability of the M2VN model, an ablation study was conducted by comparing the model's performance with and without the news modality ($N_t$). Table \ref{tab1:main_ablation} presents the out-of-sample forecasting accuracy for all seven U.S. stocks across two key metrics: Quasi-Likelihood Loss (QLike) and Mean Absolute Percentage Error (MAPE). The results unequivocally demonstrate that the inclusion of news embeddings consistently enhances forecasting performance.

Across all evaluated stocks and both metrics, the M2VN model with news embeddings (standard \texttt{M2VN}) significantly outperforms its counterpart without news embeddings (W/O $N_t$). For instance, in the case of Coca-Cola (KO), the QLike decreased from 0.0641 to 0.0589 and MAPE from 0.2904 to 0.2771 when news was incorporated. Similar improvements are observed for Comcast (CMCSA), where QLike dropped from 0.0866 to 0.0769 and MAPE from 0.3184 to 0.3025. Notably, for stocks like ConocoPhillips (COP) and Oracle (ORCL), the performance gaps are particularly pronounced, with substantial reductions in both QLike and MAPE, indicating that news content provides highly meaningful signals for these equities. For COP, QLike improved from 0.0599 to 0.0552 and MAPE from 0.3000 to 0.2722. Similarly, for ORCL, QLike went from 0.0780 to 0.0696 and MAPE from 0.3251 to 0.3092.

While the benefit of news integration is universal, the magnitude of improvement varies across stocks. For example, Gilead Sciences (GILD) shows a comparatively smaller performance gain (0.0632 vs. 0.0630 for QLike, and 0.2851 vs. 0.2838 for MAPE), suggesting that its volatility dynamics might be less influenced by publicly available news narratives, or that its price movements are predominantly driven by other factors captured within the quantitative market state. Conversely, stocks exhibiting larger gains, such as COP and ORCL, may have volatility profiles more responsive to specific news events or broader market sentiment conveyed through news. This differentiation could arise from factors like sector-specific news impact, company-specific event frequency, or the general information efficiency of the stock. News embeddings likely contribute by providing timely, qualitative insights into market-moving events or shifts in sentiment that are not immediately captured by historical price data alone. Their effectiveness in certain cases, and their less pronounced impact in others, underscore the interconnected relationship between qualitative information and asset price dynamics. Overall, the consistent outperformance of the model incorporating news embeddings provides strong evidence that this modality contributes meaningful predictive signals for volatility forecasting within the \texttt{M2VN} multi-modal architecture.

\subsection{Showcasing Forecast Quality}
To supplement the aggregate metrics in Table~\ref{tab1:main_result},  this section qualitatively compares \texttt{M2VN}’s out-of-sample forecasts against ground truth and two strong baselines — HAR-X (Lasso) and TimeXer — on the CMCSA and KO datasets (2021–2023), as shown in Figure~\ref{figure_showcase}. Notably, while all models struggle to predict the exact magnitude of extreme, single-day volatility shocks (e.g., the CMCSA spike in early 2022), \texttt{M2VN} more effectively captures these peaks compared to the baselines. In contrast, HAR-X (Lasso) and TimeXer tend to produce overly smoothed forecasts that underestimate sharp spikes, with HAR-X (Lasso) particularly weak in tracking the post-event decay. This results in persistent over-prediction during calmer periods, as seen with HAR-X (Lasso) in KO dataset from August 2022 to January 2023 and April to June 2023. Although \texttt{M2VN} may underestimate the absolute height of extreme movements, it reliably identifies the direction and timing of sharp upward shifts, while the baselines often produce muted or delayed responses. This qualitative improvement is significant, as forecasting rare volatility shocks is a well-known challenge in financial modeling due to their underrepresentation in training data, and \texttt{M2VN}'s ability to anticipate such events marks a notable advance over existing methods.

\section{Conclusion}
This study introduced \texttt{M2VN}, a novel deep learning framework that effectively unifies structured financial time series with unstructured news narratives for volatility forecasting. By leveraging temporally faithful news embeddings from TiMaGPT and cross-modal alignment through a contrastive loss, \texttt{M2VN} tackles key challenges of data modality fusion and look-ahead bias. Experiments on large-cap U.S. equities show that M2VN consistently outperforms both classical econometric models and deep learning baselines, achieving higher accuracy in both MAPE and QLike metrics. Ablation studies further confirm that both trading volume and news embeddings provide unique, complementary predictive signals. These findings highlight the value of cross-modal integration for improved volatility modeling and financial risk management.

\begin{acks}
The authors would like to thank the reviewers for their insightful and detailed comments. Y.K. gratefully acknowledges financial support from Deutsche Bank AG.
\end{acks}

\bibliographystyle{ACM-Reference-Format}
\bibliography{sample-base}










\end{document}